\documentclass[runningheads]{llncs}
\usepackage{amsmath}
\usepackage{hyperref}
\usepackage[capitalise]{cleveref}
\usepackage[T1]{fontenc}
\usepackage{graphicx}
\usepackage{color}
\usepackage{tikz}
\usepackage{thm-restate}

\usetikzlibrary{matrix,backgrounds,positioning,arrows, decorations.pathreplacing}

\usepackage{todonotes}

\newcommand*{\rlbwt}{\ensuremath{r}}
\newcommand*{\rlbbwt}{\ensuremath{r_B}}

\newcommand*{\BBWT}{\ensuremath{\mathsf{BBWT}}}
\newcommand*{\BWT}{\ensuremath{\mathsf{BWT}}}
\newcommand*{\CSA}{\ensuremath{\mathsf{CSA}}}

\newcommand{\rot}{\mathit{rot}}

\newcommand{\ignore}[1]{}
\begin{document}
\title{Bijective BWT based Compression Schemes}
\author{
    Golnaz~Badkobeh\inst{1}\orcidID{0000-0001-5550-7149}
    \and
    Hideo Bannai\inst{2}\orcidID{0000-0002-6856-5185}
    \and
    Dominik K\"oppl\inst{3}\orcidID{0000-0002-8721-4444}
}
\authorrunning{G. Badkobeh et al.}
\institute{
    City, University of London, London, UK,\\
    \email{Golnaz.Badkobeh@city.ac.uk}\\
    \and
    Tokyo Medical and Dental University, Tokyo, Japan,\\
    \email{hdbn.dsc@tmd.ac.jp}\\
    \and
    University of Yamanashi, Kofu, Japan,\\
    \email{dkppl@yamanashi.ac.jp}
}
\maketitle              \begin{abstract}
    We investigate properties of the bijective Burrows-Wheeler transform (BBWT).        		We show that for any string $w$,
    a bidirectional macro scheme of size $O(\rlbbwt{})$ can be induced from the BBWT
    of $w$,
    where $\rlbbwt{}$
    is the number of maximal character runs in the BBWT.
    We also show that
    $\rlbbwt{} = O(z\log^2 n)$, where
    $n$ is the length of $w$ and
    $z$ is the number of Lempel-Ziv 77 factors of $w$.
    Then, we show a separation between BBWT and BWT by a family of strings with $\rlbbwt{} = \Omega(\log n)$ but having only $\rlbwt=2$ maximal character runs in the standard Burrows--Wheeler transform (BWT).
    However, we observe that the smallest $\rlbbwt{}$ among all
    cyclic rotations of $w$ is always at most $\rlbwt{}$.
    While an $o(n^2)$ algorithm for computing an optimal rotation giving the smallest $\rlbbwt$ is still open,
    we show how to
    compute the Lyndon factorizations -- a component for computing BBWT --
    of all cyclic rotations in $O(n)$ time.
Furthermore, we conjecture that we can transform two strings having the same Parikh vector to each other by BBWT and rotation operations, and prove this conjecture for the case of binary alphabets and permutations.

\keywords{
        Repetitiveness measure
        \and
        Burrows--Wheeler Transform
    }
\end{abstract}
\section{Introduction}
The Burrows--Wheeler transform (BWT)~\cite{burrows94bwt} has seen numerous
applications in data compression and text indexing, and is used heavily by various tools in the field of bioinformatics.
For any string $w$, $\BWT(w)$ is defined as the string obtained
by concatenating the last characters of all
cyclic rotations of $w$, in the lexicographic order of the cyclic rotations.
$\BWT$ is not injective,
as two strings are transformed to the same string if they are cyclic rotations of each other.
Also, $\BWT$ is not surjective
since $\BWT$ preserves the string length and by the pigeonhole principle,
there exist strings that are not in the image of $\BWT$.

For a string $x$, the inverse BWT transform is induced from the
LF-mapping function $\psi_x(i)$
which maps position $i$ in $x$
to its rank among all positions ordered by $(x[i],i)$,
i.e. $\psi_x(i) = |\{ j \in [1,|x|] \mid x[j] < x[i]| + |\{ j \in[1,i] \mid x[j] = x[i]\}|$.
For a primitive string $w$,
the standard BWT always constructs a string $x = \BWT(w)$ such that
$\psi_x$ forms a single cycle (i.e., $\forall i, \exists j$ s.t. $\psi_x^j(i)=1$),
and $x[\psi_x^{|w|-1}(i)]\cdots x[\psi_x^{0}(i)]$ is a cyclic rotation of $w$.
As an example,
for $w[1..6] = \mathtt{banana}$ and $x = \BWT(w) = \mathtt{nnbaaa}$,
$\psi_x(1) = 5,
    \psi_x(2) = 6,
    \psi_x(3) = 4,
    \psi_x(4) = 1,
    \psi_x(5) = 2,
    \psi_x(6) = 3$, and thus $\psi_x^{6}(i) = i$ for all $i \in [1,6]$.

In general, $\psi_x$ can form several cycles (e.g., when $x$ is not in the image of $\BWT$),
and it is more natural to view the inverse BWT transform
as a mapping from a string to a multiset of primitive cyclic strings.
The bijective BWT (BBWT)~\cite{kufleitner09bwt,DBLP:journals/corr/abs-1201-3077,mantaci07ebwt} exploits this to define
a bijection on strings.
By selecting the lexicographically smallest rotation of each cyclic string
and concatenating them in non-increasing lexicographic order,
this mapping becomes a bijection that maps a string to another string.
The forward transform $\BBWT(w)$ can then be defined as a transform that
first computes the Lyndon factorization~\cite{chen58lyndon} of $w$,
and then taking the last symbol of
all cyclic rotations of all the Lyndon factors, sorted in $\omega$-order $\prec_\omega$,
which is an order defined, when $x,y$ are primitive, as $x \prec_\omega y \iff x^\infty \prec y^\infty$, and $\prec$ denotes the standard lexicographic order.

It is known that the BBWT can be computed in linear time~\cite{BANNAI2024105153,DBLP:conf/spire/BoucherCL0S21a}.
It can also be used as an index similar to the BWT~\cite{bannai19bbwt,BANNAI2024105153},
or as an index for a set of circular strings~\cite{BOUCHER2024105155}.
While the size $r$ of the run-length compressed BWT (RLBWT) has been a focus of study in various contexts
and is known to be small for highly-repetitive texts~\cite{navarro21indexing1},
the size $r_B$ of the run-length compressed BBWT (RLBBWT) has not yet been studied rigorously.
Biagi et al.~\cite{DBLP:conf/ictcs/0002CLR23} study the sensitivity~\cite{DBLP:journals/iandc/AkagiFI23} of $r_B$ with respect to the {\it reverse}
operation, and present an infinite family of strings such that $r_B$ of a string and its reverse can differ by a factor of $\Omega(\log n)$.

In this paper, we investigate properties of
BBWT and $\rlbbwt{}$.
In detail, we show that we can induce a bidirectional macro scheme (BMS)~\cite{DBLP:journals/jacm/StorerS82} of size $O(\rlbbwt{})$ for $w$,
from the RLBBWT of $w$~(\cref{lemInduceBMS}).
We further show that $\rlbbwt{} = O(z\log^2 n)$ where $z$ is the number of Lempel-Ziv 77 (LZ77) factors of $w$ (\cref{thmUpperBoundZ}).
Then, we show a separation between $r_B$ and $r$,
by a family of strings with $\rlbbwt{} = \Omega(\log n)$ but $r = 2$~(\cref{thmLognMultBWT}).

Noticing that $\rlbbwt{} = \rlbwt{}$ for Lyndon words,
the smallest $\rlbbwt{}$ among all cyclic rotations of
$w$ is always at most $\rlbwt$.
While we do not yet know how to compute, in subquadratic time, such an optimal rotation that gives the smallest $r_B$, we
show that we can compute the Lyndon factorizations of all rotations of $w$ in linear time (\cref{thmCyclicLyndonFact}).

Finally, we conjecture that two strings having the same Parikh vector can be transformed to each other by BBWT and rotation operations (Conjecture~\ref{conjecture:reachability}); we prove this conjecture for special cases (\cref{cor:reachable}).

\section{Preliminaries}
Let $\Sigma$ be a set of symbols referred to a the {\em alphabet},
and $\Sigma^*$ the set of strings over $\Sigma$.
For a string $x\in\Sigma^*$, $|x|$ denotes $x$'s length.
The empty string (the string of length $0$) is denoted by $\varepsilon$.
For integer $i \in [1,|x|]$, $x[i]$ is the $i$th symbol of $x$,
and for integer $j \in[i,|x|]$, $x[i..j] = x[i]\cdots x[j]$.
For convenience, let $x[i..j] = \varepsilon$ if $i > j$.
Let $x = x^1$, and for integer $k\geq 2$, $x^k = x x^{k-1}$
A string is {\em primitive},
if it cannot be represented as $x^k$ for some string $x$ and integer $k\geq 2$.

Let $\rot(x) = x[|x|]x[1..|x|-1]$.
A string $y$ is a {\em cyclic rotation} (or simply a {\em   rotation}) of $x$
if there exists $i$ such that $y = \rot^i(x)$,
where
$\rot^1(x)=\rot(x)$, and for integer $k \geq 2$,
$\rot^k(x) = \rot^{k-1}(\rot(x))$.

Given a total order $\prec$ on $\Sigma$,
the lexicographic order (also denoted by $\prec$)
induced by $\prec$ is a total order on $\Sigma^*$ such that
$x\prec y$ if and only if $x$ is a prefix of~$y$, or, $x[i] \prec y[i]$ where $i = \min \{ k \geq 1 \mid x[k]\neq y[k] \}$.
A string $w$ is a {\em Lyndon word}, if it is lexicographically smaller than all of its proper suffixes~\cite{lyndon54}.
Lyndon words must therefore be primitive.
Also, any string $w$ can be partitioned into a
unique sequence of lexicographically non-increasing Lyndon words, called the {\em Lyndon factorization}~\cite{chen58lyndon} of $w$,
i.e., $w = f_1^{k_1}\cdots f_{\ell(w)}^{k_{\ell(w)}}$ where each $f_i~(1\leq i \leq \ell(w))$ is a Lyndon word, and $f_{i} \succ f_{i+1}$ for all $1 \leq i < \ell(w)$.
We call $f_i^{k_i}$ the $i$-th \emph{Lyndon necklace} of $w$.
The $\omega$-order $\prec_\omega$ is a total order over primitive strings,
defined as: $x\prec_\omega y$ if and only if $x^\infty\prec y^\infty$.
\footnote{Mantaci et al.~\cite{mantaci07ebwt} define the $\omega$-order as a total order over arbitrary strings (including non-primitive strings), but as it is not relevant in our presentation, we omit this for simplicity.}

Given a string $w$,
the Burrows--Wheeler transform $\BWT(w)$ is a string obtained by
concatenating the last symbol of all cyclic rotations of $w$, in lexicographic order.
The bijective BWT $\BBWT(w)$ is a string obtained by concatenating the last symbol
of all cyclic rotations of all Lyndon factors in the Lyndon factorization of $w$,
in $\omega$-order.
The number of maximal same-character runs in $\BWT(w)$ and $\BBWT(w)$ will be denoted by
$\rlbwt{}(w)$, and $\rlbbwt{}(w)$ respectively.
Although $\rlbwt{}(w)$, $\rlbbwt{}(w)$ and $\ell(w)$ are functions on strings to non-negative integers, we will omit writing the considered string and just write $\rlbwt{},\rlbbwt{},\ell$, if the context is clear.

\section{Properties of $\rlbbwt{}$}

We here analyze $\rlbbwt{}$ as a repetitiveness measure for a string~$w$.
We first confirm that $\rlbbwt{}$ corresponds to the size of a \emph{bidirectional macro scheme (BMS)},
and is a repetitiveness measure for a form of dictionary compression, as is RLBWT.
A BMS~\cite{DBLP:journals/jacm/StorerS82},
the most expressive form of dictionary compression,
partitions $w$ into phrases, such that each phrase
of length at least 2 can be represented as a reference to another substring of~$w$.
The referencing of the phrases induces a referencing forest over the positions:
any position in a phrase of length at least 2 references another position in~$w$,
such that all positions in the same phrase have the same offset and
thus adjacent positions point to adjacent positions,
and there are no cycles.

\begin{lemma}\label{lemInduceBMS}
    There exists a BMS of size $O(\rlbbwt(w))$ that represents the string~$w$.
\end{lemma}
\begin{proof}
    We follow the existence proof for a BMS of size $O(r(w))$ by
Navarro et al.~\cite{navarro21approximation}.
    We consider a BMS such that
    each text position that does not correspond to a beginning of a same-character run in $\BBWT(w)$ will reference the text position
    corresponding to the preceding character in the $\BBWT(w)$.
    It is clear that there are no cycles in such a referencing,
    and we claim that this allows the string to be partitioned into
    $O(\rlbbwt{})$ phrases, such that
    references of adjacent positions in a given phrase of length at least two point to adjacent positions.

    Focus on the $i$-th Lyndon necklace $f_i^{k_i}$ of $w$.
    For any cyclic rotation of $f_i$,
    its $k_i$ copies occur adjacently when writing all cyclic conjugates of all Lyndon factors in the $\omega$-order,
    and thus correspond to adjacent characters in a run in $\BBWT(w)$.
    It follows that for any position in the last $k-1$ copies of $f_i$, the reference points to the corresponding position in the preceding copy. Thus, adjacent positions reference adjacent positions, and can be contained in the same phrase.

    Next, consider a position in the first copy of $f_i$ that does not correspond to a beginning of a same-character run in $\BBWT(w)$.
    Then, since the character at this position and the preceding (in $\omega$-order) cyclic string is the same,
    their preceding positions must also correspond to adjacent positions in $\BBWT(w)$.
    This implies that,
    as long as the corresponding position is again not a beginning of a same-character run,
    adjacent positions will refer to adjacent positions, albeit, in the cyclic sense.
    Being the beginning of a same-character run in $\BBWT(w)$ can happen at most $\rlbbwt{}$ times.
    Being adjacent in the cyclic sense, but not being adjacent in text-order can happen at most once
    per referenced Lyndon necklace.
    Thus, the number of times adjacent text positions can be in a different phrase is bounded by $O(\rlbbwt{}+\ell(w))$.
    Since $\ell(w) \leq \rlbbwt{}$ can be shown from Corollary 2. of~\cite{boucher21rindexing},
    this concludes the proof.
\qed
\end{proof}

\begin{example}\label{exBMS}
    \def\checkmark{\tikz\fill[scale=0.3](0,.35) -- (.25,0) -- (1,.7) -- (.25,.15) -- cycle;}
    For the string $w =\texttt{abbbabbababab} $, we give below an example for the BMS computed from its BBWT\@.
    The Lyndon factorization of $w$ is \texttt{abbb}, \texttt{abb}, \texttt{ab}, \texttt{ab}, \texttt{ab}.
    The number of runs $\rlbbwt$ is $6$.
    BBWT positions belonging to a referencing phrase are marked with \checkmark{} in the last row of the table below.
    We therefore have $6$ non-referencing phrases.

    \begin{center}
        \begin{tabular}{l*{13}{p{1.5em}}}
            $i$          & 1 & 2          & 3          & 4          & 5          & 6 & 7          & 8          & 9 & 10         & 11 & 12 & 13 \\
            $w[i]$       & a & b          & b          & b          & a          & b & b          & a          & b & a          & b  & a  & b  \\
            $\BBWT[i]$   & b & b          & b          & b          & b          & a & a          & a          & b & b          & a  & b  & a  \\
            $\CSA[i]~$   & 8 & 10         & 12         & 5          & 1          & 9 & 11         & 13         & 7 & 4          & 6  & 3  & 2  \\
            $\CSA[i]-1~$ & 9 & 11         & 13         & 7          & 4          & 8 & 10         & 12         & 6 & 3          & 5  & 2  & 1  \\
            ref?         &   & \checkmark & \checkmark & \checkmark & \checkmark &   & \checkmark & \checkmark &   & \checkmark &    &    &    \\
        \end{tabular}
    \end{center}

    The referencing phrases are computed as follows:
    In the table above, $\CSA$ is the circular suffix array whose entry $\CSA[i]$ denotes the text position after the one from which we took $\BBWT[i]$ (or, if $\BBWT[i]$ belongs to the last character of a Lyndon factor $F$, the starting position of $F$).
    The row $\CSA[i]-1$ denotes the text position corresponding to $\BBWT[i]$.
    By construction of our BMS, the \CSA{} entry positions 2-5, 7-8, and 10 correspond to referencing phrases, i.e.,
    the text positions after 10,12,5,1,11,13,4 (applying \CSA{} on these entry positions), i.e.,
    the text positions
    11, 13, 7, 4, 10, 12, 3.
    These positions refer to
    9, 11, 13, 7, 8, 10, 6, respectively.
    If we group together neighboring positions that have the same offset to its reference,
    we obtain $9$ phrases, which are visualized in \cref{fig:bmsfact}.
\end{example}

\begin{figure}[htpb]
    \centering
    \includegraphics{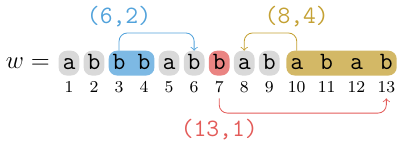}
    \caption{The BMS factorization of \cref{exBMS}}
    \label{fig:bmsfact}
\end{figure}

\begin{theorem}\label{thmUpperBoundZ}
    For any string, $\rlbbwt{} = O(z\log^2 n)$.
\end{theorem}
\begin{proof}
    (Sketch)
    We follow the proof of $r = O(z\log^2 n)$ by Kempa and Kociumaka~\cite{DBLP:journals/cacm/KempaK22}.
    The LCP array of a string $w$ is an array of integers such that its $i$-th entry is the length of the longest common prefix ($lcp$) between
    the lexicographically $(i-1)$-th and $i$-th cyclic rotation of $w$.
    An \emph{irreducible LCP position} is a position $i$ such that $i=1$ or $\BWT(w)[i-1]\neq\BWT(w)[i]$, and thus
    the number $r$ of BWT runs is the number of irreducible LCP positions.
    For a multiset of primitive cyclic strings,
    we analogously define
    the $\omega$-LCP array
    such that its $i$-th entry is
    the $lcp$ between $x^\infty$ and $y^\infty$,
    where $x$ and $y$ are respectively the $(i-1)$-th and $i$-th string, in $\omega$-order,
    among all cyclic rotations of all Lyndon factors of the Lyndon factorization of $w$.
    By construction, the number $\rlbbwt{}$ of BBWT runs is the number of irreducible $\omega$-LCP positions,
    i.e., $i=1$ or $\BBWT(w)[i-1]\neq\BBWT(w)[i]$.
    Note that $\omega$-LCP values can be infinite when there are Lyndon necklaces in the Lyndon factorization with exponent at least~$2$, but they can be safely disregarded since they are not irreducible.
    The theorem follows if we can
    show that for any value $k$, the number of irreducible
    $\omega$-LCP values in $[k,2k)$ is $O(z\log n)$
    and considering $k = 2^i$ for $i = 0,\ldots, \lfloor \log n\rfloor$.

    The arguments in the proof in~\cite{DBLP:journals/cacm/KempaK22} proceed by first
    asserting that for any integer $k$, a string contains at
    most $3kz$ distinct strings of length $3k$. Then, each irreducible
    LCP value in $[k,2k)$ is associated with a cost of $k$, which are charged to positions
    in the at most $3kz$ strings that have an occurrence crossing the corresponding suffix array position, and it is shown that each substring can be charged at most $2\log n$ times.
    The total cost is thus at most $6kz\log n$ and thus the number of irreducible LCP values is $O(z\log n)$.

    For $\omega$-LCP, the corresponding length $3k$ substring associated with the suffix array position may not occur in the original string but instead will correspond to a substring of some Lyndon necklace of the Lyndon factorization.
    Note that there are at most $3k$ distinct substrings of length $3k$ that are not substrings of the original string but a substring of a given Lyndon necklace.
    Since $\ell(w)<4z$~\cite{DBLP:conf/cpm/UrabeNIBT19},
    we have that the total number of such distinct substrings of length $3k$ that occur
    in this context is still bounded by $O(kz)$, and that the arguments still hold.\qed
\end{proof}

Despite sharing common traits,
$\rlbwt{}$ and $\rlbbwt$ can be
asymptotically different:

\begin{theorem}\label{thmLognMultBWT}
    There exists a family of strings with
    $\rlbbwt{} = \Omega(\log n)$ and $\rlbwt = 2$.
\end{theorem}
\begin{proof}
Define the Fibonacci words as follows:
    $F_0 = \texttt{b}, F_1 = \texttt{a}, F_i = F_{i-1}F_{i-2}$.
    The infinite Fibonacci word is $\lim_{k\rightarrow\infty} F_k$.
    Melan\c{c}on~\cite{melancon99lyndon}
    showed that the $k$-th factor (the first factor being the $0$-th)
    of the Lyndon factorization of the infinite Fibonacci word,
    has length $f_{2k+2}$, where $f_i = |F_i|$.
    Now, $\sum_{k=0}^{i} f_{2k+2} =
        -f_1 + (\cdots(((f_1 + f_2) + f_4) + f_6) + \cdots) + f_{2i+2}
        = f_{2i+3}-1$.
    Therefore, the word obtained by deleting the last symbol of $F_{2i+3}$
    has $i+1$ distinct Lyndon factors.
    Noticing that the last symbol of $F_{2i+3}$ must be `$\texttt{a}$' and will
    form a distinct
    Lyndon factor, we have that the size of the Lyndon factorization of $F_{2i+3}$ is $i+2$.
    Since $\ell(w) \le \rlbbwt$~\cite{boucher21rindexing},
the \BBWT{} of the $k$-th Fibonacci word~$F_k$ for odd $k$
    has $\Omega(k)$ runs,
while the BWT of any Fibonacci word has $\rlbwt(F_k) = 2$ runs~\cite{mantaci03sturmian}.
\qed
\end{proof}

We have not yet been able to find a family of strings
where $\rlbbwt{} = o(\rlbwt{})$.

\section{RLBBWT and Rotation}
Theorem~\ref{thmLognMultBWT} may give the impression that $\rlbwt{}$ may be a
smaller measure compared to~$\rlbbwt{}$.
However, if we are to incorporate a rotation operation, which can be encoded as a single
$\log n$-bit integer, we could possibly obtain a representation smaller than $\rlbwt{}$ using BBWT.
This is because we have $\rlbwt(x) = \rlbbwt(x)$ for the Lyndon rotation~$x$ of any primitive word.
The Lyndon rotation of a primitive word can be computed in linear time~\cite{shiloach81fast}.

For any string $w$, let $\hat{w} = \arg\min_{uv = w}\{ \rlbbwt(vu) \}$
be the {\em optimal} rotation with respect to~$\rlbbwt$.
We observe that $\hat{w}$ is not always the Lyndon rotation of $w$.
For example, for the Lyndon word
$w=\texttt{aaabaabaaabaabb}$,
we have that
$\BWT(w) = \BBWT(w) = \texttt{bbbaabaaaabaaaa}$,
thus $\rlbbwt(w) = 6$.
However, we have that
$\hat{w} = \rot(w) = \texttt{baaabaabaaabaab}$
and $\BBWT(\hat{w}) = \texttt{bbbbaaaaaaaaaab}$,
thus $\rlbbwt(\hat{w})=3$.

Since $\BBWT(w)$ (and hence $\rlbbwt(w)$) can be computed in $O(n)$ time,
it is straightforward to compute $\hat{w}$ (and hence $\rlbbwt(\hat{w})$) in $O(n^2)$ time.
A subquadratic time algorithm for this problem would be very interesting.
(For LZ77, it was recently shown that indeed subquadratic time computation is possible~\cite{DBLP:conf/cpm/BannaiCR24}.)
While we have not yet been able to achieve this,
we give a partial result: a linear time algorithm for computing
the Lyndon factorizations, a precursor to computing BBWT,
of all cyclic rotations.
\begin{theorem}\label{thmCyclicLyndonFact}
    We can compute the sizes of the Lyndon factorizations of all cyclic rotations of $w$ in time linear in the length of $w$.
\end{theorem}
\begin{proof}
(Sketch)
    Assume that $w$ is Lyndon,
    consider the string $W=ww$,
    and view the cyclic rotations of $w$ as substrings of length $|w|$ of $W$.
    Any Lyndon factorization of such a substring consists
    of the Lyndon factorization of a suffix of $w$ and a prefix of $w$,
    since any $x = uv$ such that $u$ is a suffix of $w$ and $v$ is a prefix of $w$ cannot be Lyndon:
    it would imply $uv \prec v \prec w \prec u$, a contradiction.

    We observe that the factors of the Lyndon factorization for any suffix of~$w$,
    are the sequence of maximal right sub-trees of the standard (right) Lyndon tree~\cite{bannai17runs} of $w$ that are contained in the suffix.
    Similarly, for any prefix of $w$, they are the maximal left sub-trees of the left Lyndon tree~\cite{DBLP:journals/iandc/BadkobehC22} of $w$ that are contained in the prefix.

    The right Lyndon tree of a Lyndon word $w$ is a binary tree
    defined recursively as follows:
    if $w$ is a single letter, it is a leaf,
    otherwise, the left and right child are respectively the
    right Lyndon trees of $u,v$
    where $w=uv$ and $v$ is the longest proper suffix of $w$ that is a Lyndon word. Note that it can be shown that this choice of $v$ implies that $u$ is a Lyndon word.
    The left Lyndon tree is defined analogously, but
    $u$ is the longest proper prefix of $w$ that is a Lyndon word.
    Similarly, it can be shown that this choice of $u$ implies that $v$ is a Lyndon word.

    By definition of the Lyndon trees, and the property of the Lyndon factorization
    which states that the first (resp. last) factor is the longest prefix (resp. suffix) that is a Lyndon word,
    it is a simple observation that the Lyndon factorization of a suffix of $w$
    is exactly the sequence of maximal right nodes of the right Lyndon tree
    that are contained in the suffix,
    and the Lyndon factorization of a prefix of $w$ is exactly
    the sequence of
    maximal left nodes of the left Lyndon tree
    that are contained in the prefix. See Fig~\ref{fig:lyndontrees} for an example.

    Both trees can be computed in linear time~\cite{bannai17runs,DBLP:journals/iandc/BadkobehC22}.
    It is not difficult to see that the changes in the sequences,
    and thus the sizes of the Lyndon factorizations
    can be computed in total linear time for each of the suffixes and prefixes,
    by a left-to-right traversal on the trees.

    \begin{figure}[h]
        \centerline{
            \includegraphics[width=0.8\textwidth]{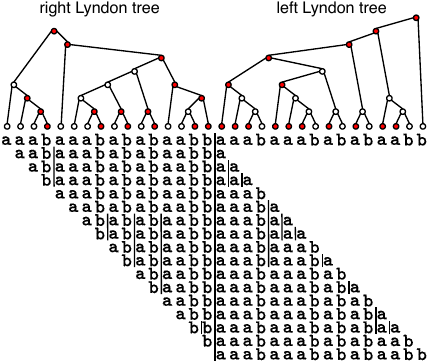}
        }
        \caption{The left and right Lyndon trees of the string
            $w = \texttt{aaabaaabababaabb}$.
            The Lyndon factorization of cyclic rotations of $w$ are shown below,
            where factors are delimited by vertical bars.
            Right-nodes of the right Lyndon tree, and left-nodes of the left Lyndon tree are marked in red.
}\label{fig:lyndontrees}
    \end{figure}

\qed
\end{proof}

\section{BBWT Reachability}
Further developing the idea of combining rotation and \BBWT{}
in order to obtain a smaller representation of a string, we give the following conjecture.
\begin{conjecture}\label{conjecture:reachability}
    Given two words with the same Parikh vector,
    we can transform one to the other by using only rotation
    and \BBWT{} operations.
\end{conjecture}
If true, this would suggest that it is possible, for example, to represent a string~$w$ based on $w$'s Parikh vector,
and a sequence of integers where each integer alternately represents the offset of the rotation or the number of times $\BBWT{}$ is applied,
to reach $w$ from the lexicographically smallest string with the same Parikh vector.

We have computationally confirmed the conjecture for ternary strings of up to length 17, with code available at \url{https://github.com/koeppl/bbwtreachability},
and have proved it for the specific
cases where the alphabet is binary, or, when all symbols are distinct.
We first give the following lemma which shows that, under the condition of the lemma,
we can obtain a lexicographically smaller string using
rotation operations and an inverse BBWT operation.

\begin{lemma}\label{lemma:Lyndon_oneshift_smaller}
    Let $x$ be a word of length $n$ whose smallest rotation is itself (i.e., a necklace),
    over the alphabet
    $\{ c_1,\ldots, c_\sigma \}$
    where  $c_1 \prec \cdots \prec c_\sigma$,
    and let $(e_1,\ldots, e_\sigma)$ be the Parikh vector of $x$.
    Let $y = c_1^{e_1}\cdots c_\sigma^{e_\sigma}$,
    i.e., the lexicographically smallest string with the same Parikh vector as $x$,
    and $i = lcp(x,y)$. If $x[i] \neq x[n]$, then,
    there exists $k$ such that
    $\rot^k(\BBWT^{-1}(\rot(x))) \prec x$.
\end{lemma}
\begin{proof}
    Note that $x[i] \neq x[n]$ implies $y \prec x$ since $y = x$ implies $i=n$.
    Let $x[1..i] = c_1^{e_1}\cdots c_k^{e'}$, where $e' \leq e_k$.
    This implies that symbols smaller than $c_{k}$ are all used up in $x[1..i] = y[1..i]$,
    and cannot occur in $x[i+1..n]$ nor $y[i+1..n]$.
    Thus,
    for all $j\in[i+1,n]$,
    it holds that
    $x[j] \succeq x[i] = c_k$,
    in particular,
    for all $j\in[1,i]$,
    it holds that
    $x[n] \succ x[i] \succeq x[j]$ since $x[i] \neq x[n]$ is assumed.

    Next, consider traversing the symbols of $\hat{x} = \rot(x) = x[n]x[1..n-1]$ using the LF mapping $\psi_{\hat{x}}$
    starting from position $i'$ such that $\psi_{\hat{x}}(i') = i+1$ to recover a cyclic substring of $\BBWT^{-1}(\hat{x})$,
    whose smallest rotation (Lyndon rotation)
    will be a substring of $\BBWT^{-1}(\hat{x})$.
    Thus, we start from the symbol $\hat{x}[i']=y[i+1]$.
    Since $\hat{x}[1] = x[n] \neq x[j]=y[j]$ for any $j\in [1..i]$
    and $\hat{x}[2..i+1]=x[1..i] = y[1..i]$, it follows that
    $\psi_{\hat{x}}(j) = j-1$ for any $j\in [2..i+1]$.
    Therefore, we have that $y[i+1]$ is prefixed by $x[1..i]$, i.e.,
    $x[1..i]y[i+1] \prec x[1..i+1]$ is a cyclic substring of $\BBWT^{-1}(\hat{x})$,
    where the inequality follows from the definition of $y$ and $i$.
    Since the length $i+1$ prefix of the Lyndon rotation
    of the whole cyclic substring that is retrieved by $\psi_{\hat{x}}$ starting from $i+1$ cannot be larger than $x[1..i]y[i+1]$,
    it follows that
    $\BBWT^{-1}(\hat{x})$ contains a length $i+1$ substring that is smaller than $x[1..i+1]$, and the lemma holds.
    \qed
\end{proof}

\begin{example}\label{ex:Lyndon_oneshift_smaller}
    For $x = \mathtt{aacb}$,
    $\BBWT^{-1}(x[n] \cdot x[1..n-1])
        = \BBWT^{-1}(\mathtt{baac})
        = \mathtt{caab}$.
    The smallest rotation of $\mathtt{caab}$ is $\mathtt{aabc}$, which is lexicographically smaller than $x$.
\end{example}
In \cref{ex:Lyndon_oneshift_smaller},
$\mathtt{caab}$ coincides here with the string $y$ in the proof of \cref{lemma:Lyndon_oneshift_smaller}.
A schematic sketch of this proof is given in \cref{fig:Lyndon_oneshift_smaller}.

\begin{figure}[htpb]
    \centering
    \begin{tikzpicture}
        \tikzset{position label/.style={above = 3pt,
},
            brace/.style={decoration={brace},
                    decorate
                }
        }
        \tikzstyle{array} = [matrix of nodes,font=\ttfamily, column sep=0.5\pgflinewidth, row sep=0.5mm, nodes in empty cells,
row 1/.style={nodes={font=\scriptsize,fill=none, minimum size=0mm, text height=0.33em}},
        ]
        \matrix[array] (M) {$v$             & 1      & 2        & \ldots   &          &          &          &  &  &  &  &  &  &  & \ldots & $i$ & $i+1$
            \\
            $x[v] = $       & $c_1$  & $c_1$    & $\ldots$ & $c_1$    &
            $c_2$           & $c_2$  & $\ldots$ & $c_2$    &
            $c_3$           & $c_3$  & $\ldots$ &
            $c_k$           & $c_k$  & $\ldots$ & $c_k$    & $\beta$  & $\ldots$
            \\
            \\
            $y[v] = $       & $c_1$  & $c_1$    & $\ldots$ & $c_1$    &
            $c_2$           & $c_2$  & $\ldots$ & $c_2$    &
            $c_3$           & $c_3$  & $\ldots$ &
            $c_k$           & $c_k$  & $\ldots$ & $c_k$    & $\alpha$ & $\ldots$
            \\
            \\
            $\hat{x}[v] = $ & $x[n]$ & $c_1$    & $\ldots$ & $c_1$    &
            $c_1$           & $c_2$  & $\ldots$ & $c_2$    &
            $c_2$           & $c_3$  & $\ldots$ &
            $c_{k-1}$       & $c_k$  & $\ldots$ & $c_k$    & $c_k$    & $\beta$  & $\ldots$
            \\};
        \draw [->] (M-6-3) -- (M-4-2.south) -- node [midway,anchor=east] {$\psi$} (M-6-2);
        \draw [->] (M-6-4) -- (M-4-3.south) -- (M-6-3);
        \draw [->] (M-6-5) -- (M-4-4.south) -- (M-6-4);
        \draw [->] (M-6-6) -- (M-4-5.south) -- (M-6-5);
        \draw [->] (    M-6-7) -- (M-4-6.south) -- (M-6-6);
        \draw [->] (M-6-8) -- (M-4-7.south) -- (M-6-7);
        \draw [->] (M-6-9) -- (M-4-8.south) -- (M-6-8);
        \draw [->] (M-6-10) -- (M-4-9.south) -- (M-6-9);
        \draw [->] (M-6-11) -- (M-4-10.south) -- (M-6-10);
        \draw [->] (M-6-12) -- (M-4-11.south) -- (M-6-11);
        \draw [->] (M-6-13) -- (M-4-12.south) -- (M-6-12);
        \draw [->] (M-6-14) -- (M-4-13.south) -- (M-6-13);
        \draw [->] (M-6-15) -- (M-4-14.south) -- (M-6-14);
        \draw [->] (M-6-16) -- (M-4-15.south) -- (M-6-15);
        \draw [->] (M-6-17) -- (M-4-16.south) -- (M-6-16);

        \draw [brace,decoration={raise=0ex}] (M-4-2.north west) -- node [position label] {$e_1$} (M-4-2.north -| M-4-5.north east);
        \draw [brace,decoration={raise=0ex}] (M-4-6.north west) -- node [position label] {$e_2$} (M-4-6.north -| M-4-9.north east);

    \end{tikzpicture}
    \caption{Schematic sketch of the proof of \cref{lemma:Lyndon_oneshift_smaller}. Here,
    $\alpha = y[i+1]$ and $\beta = x[i+1]$ with
    $c_k \leq y[i+1] = \alpha < x[i+1] = \beta$.
    Repeating the LF mapping
    produces a string that contains $x[1..i]y[i+1]$.
    }
    \label{fig:Lyndon_oneshift_smaller}
\end{figure}

The following Theorem follows from Lemma~\ref{lemma:Lyndon_oneshift_smaller}.
\begin{theorem}\label{cor:reachable}
    Given two words of the same length with the same Parikh vector,
    it is possible to transform one to the other by using only rotations and \BBWT{} transformations
    if all symbols are distinct, or
    if the alphabet is binary.
\end{theorem}
\begin{proof}
    Given any word, consider its smallest rotation $x$, and let $y$ be the smallest word with the same Parikh vector.
    Since \BBWT{} and rotations are bijections,
    it is easy to see that $\BBWT^{-1}(x)$ (resp. $\rot^{-1}(x)$) can be represented
    by a sequence of $\BBWT(x)$ (resp. $\rot(x)$) operations.
    Therefore, it suffices to show that we can reach $y$ from $x$ using any of these operations.
    If $y \prec x$,
    using Lemma~\ref{lemma:Lyndon_oneshift_smaller}, we can always obtain a strictly lexicographically smaller string using rotations and $\BBWT^{-1}$ and thus eventually reach $y$:
    when all symbols are distinct, it is easy to see that the condition of Lemma~\ref{lemma:Lyndon_oneshift_smaller} holds.
    If the alphabet is binary, i.e., $\{\texttt{a},\texttt{b}\}$,
    we have that $x[n] = \texttt{b}$ since $x$ is a smallest rotation.
    Furthermore,
    if $i = lcp(x,y)$, then, since $x[i]= \texttt{b}$ would imply $x = y$,
    we have $x[i] = \texttt{a} \neq \texttt{b} = x[n]$.
    \qed
\end{proof}

\section*{Acknowledgments}
This work was supported by JSPS KAKENHI Grant Numbers JP24K02899 (HB) and JP23H04378 (DK).

\clearpage
\bibliographystyle{splncs04}
\bibliography{refs,literature}

\end{document}